\documentclass[10pt,a4paper]{article}
\usepackage{amssymb}
\usepackage{amsmath}
\usepackage{amsthm}
\usepackage{latexsym}
\usepackage[dvips]{epsfig}
\usepackage{enumerate}
\usepackage{mathrsfs}
\usepackage{eufrak}
\usepackage{bm}
\usepackage{tikz}
\usepackage{authblk}
\usepackage{cancel}
\usepackage{cleveref}

\theoremstyle{plain}
\newtheorem{proposition}{Proposition}
\newtheorem{lemma}{Lemma}
\newtheorem{theorem}{Theorem}

\newtheorem{remark}{Remark}

\def\bmg{{\bm g}}

\def\bmi{{\bm i}}
\def\bmj{{\bm j}}

\newcounter{mnotecount}

\newcommand{\mnotex}[1]
{\protect{\stepcounter{mnotecount}}$^{\mbox{\footnotesize $\bullet$\themnotecount}}$ 
\marginpar{
\raggedright\tiny\em
$\!\!\!\!\!\!\,\bullet$\themnotecount: #1} }

\begin{document}

\title{\textbf{ Killing boundary data for anti-de Sitter-like spacetimes
}}

\author[,1]{Diego A. Carranza  \footnote{E-mail address:{\tt d.a.carranzaortiz@qmul.ac.uk}}}
\author[,1]{Juan A. Valiente Kroon \footnote{E-mail address:{\tt j.a.valiente-kroon@qmul.ac.uk}}}
\affil[1]{School of Mathematical Sciences, Queen Mary University of London,
Mile End Road, London E1 4NS, United Kingdom.}

\maketitle

\begin{abstract}
Given an initial-boundary value problem for an anti-de Sitter-like
spacetime, we analyse conditions on the conformal boundary ensuring
the existence of Killing vectors in the arising spacetime. This
analysis makes use of a system of conformal wave equations describing
the propagation of the Killing equation first considered by Paetz. We
identify an obstruction tensor constructed from Killing vector
candidate and the Cotton tensor of the conformal boundary whose
vanishing is a necessary condition for the existence of Killing
vectors in the spacetime. This obstruction tensor vanishes if the
conformal boundary is conformally flat.
\end{abstract}

\section{Introduction}

Anti-de Sitter-like spacetimes are solutions to the Einstein field
equations with negative Cosmological constant having a global
structure similar to that of the anti-de Sitter spacetime. In particular,
they can be conformally extended in such a way that the resulting
conformal boundary is a timelike hypersurface of the conformal
extension. Members of this class of solutions to the Einstein field
equations constitute prime examples of spacetimes
which are not \emph{globally hyperbolic}. Accordingly, initial data is
not enough to reconstruct one of these solutions to the Einstein field
equations ---one also needs to prescribe some suitable data at the
conformal boundary. The construction of anti-de Sitter spacetimes by means of a
initial-boundary value problem has been analysed in \cite{Fri95} where
a large family of \emph{maximally dissipative boundary conditions}
involving incoming and outgoing components of the Weyl tensor have
been identified. In this respect,  anti-de Sitter spacetimes provide a convenient
setting to study initial-boundary value problems for the Einstein
equations as the conformal boundary is a hypersurface with a rich
structure ---despite the use of the conformal Einstein field
equations, the formulation of the initial-boundary value problem for
anti-de Sitter-like spacetimes as given in \cite{Fri95} is
considerably simpler than the analysis of the \emph{general}
initial-boundary value problem for the Einstein field equations as
given in e.g. \cite{FriNag99}. In particular, the anti-de Sitter
construction allows to establish \emph{geometric uniqueness} while the
analysis in \cite{FriNag99} leaves unanswered this question ---see
\cite{Fri09} for a further discussion on this important issue. 

\medskip
The problem of encoding (continuous) symmetries of a spacetime at the
level of initial data is an important classical problem in Relativity
---see e.g. \cite{Mon75}. A modern presentation of this issue and the
related theory can be found in \cite{Chr91b,BeiChr97b}. The key
outcome of this theory is the so-called set of \emph{Killing initial data
equations}, a system of overdetermined equations
for a scalar field and a spatial vector on a spacelike
hypersurface ---corresponding, respectively, to the \emph{lapse} and \emph{shift}
with respect to the normal of the hypersurface of an hypothetical
Killing vector of the spacetime.  If these Killing equations admit a
solution, a so-called \emph{Killing initial data set (KID)}, then the
development of the initial data will have a Killing vector. The
theory of KID for the Cauchy problem for the Einstein field equations
can be also adapted to other settings like the (finite and asymptotic
) characteristic initial value problem \cite{ChrPae13a,Pae14} and,
more relevant for the purposes of the present article, to the
asymptotic initial value problem for the de Sitter-like spacetimes
\cite{Pae16} ---i.e. solutions to the vacuum Einstein field equations
with positive Cosmological constant.

\subsection*{Main results of the present article}
The purpose of the present article is to construct a theory of Killing
initial and boundary data in the setting of anti-de Sitter-like spacetimes. Given
the nature of the problem, we perform the analysis in a
\emph{conformal setting} ---that is, we work with a suitable
(\emph{unphysical}) conformal representation of the spacetime rather
than with the \emph{physical} spacetime itself. As these spacetimes
are not globally hyperbolic, in addition to
satisfying the KID equations on some initial hypersurface, one also
needs to prescribe some \emph{Killing boundary data (KBD)} to ensure
the existence of a Killing vector in the spacetime. The use of a
conformal setting allows to perform the analysis of the boundary
conditions for the Killing equations by means of local (differential
geometric) computations. The Killing boundary data restricts, in turn,
the structure of the conformal boundary. In addition, the Killing
initial and boundary data have to satisfy some compatibility
conditions at the corner where the initial hypersurface and the
conformal boundary meet.

\medskip
Our strategy to identify the Killing boundary data is to make use of a
system of conformal wave equations describing the propagation of the
Killing vector equation first discussed by Paetz in \cite{Pae16}
---the \emph{Killing equation conformal propagation system}, see 
\Cref{WEZQ}, equations \eqref{KillingEvolution1}-\eqref{KillingEvolution5}. If this
system has the trivial (vanishing) solution then a suitably
constructed Killing vector candidate is, in fact, a Killing vector of
the spacetime. Accordingly, one is naturally lead to consider an
initial-boundary value problem with both vanishing initial data and
Dirichlet boundary data for the Killing equation conformal propagation
system. While the vanishing initial data naturally leads to a
conformal version of the Killing initial data equations, the vanishing
Dirichlet boundary data give the Killing boundary data conditions
---see equations \eqref{PhiBdy}-\eqref{BabBdy}. A detailed formulation
of this result is given in \Cref{VanishingZQ}. The conditions
obtained by this approach are, in first instance, restrictions on
spacetime tensors. In a second step, we analyse the
interdependencies between these conditions and express them in terms
of objects which are intrinsic to the conformal boundary --- the
 \emph{reduced Killing boundary equations}, equations \eqref{EtaBdy}-\eqref{DerDerEtaBdy}. A key ingredient in
this analysis is given by the constraint equations, \eqref{CCCB1}-\eqref{CCCB10}, implied by the conformal
Einstein equations on the timelike conformal boundary.

The analysis of the reduced Killing boundary equations shows that a
necessary condition for the existence of a Killing vector in the
anti-de Sitter-like spacetime is the existence of a conformal Killing
vector in the conformal boundary ---see equation
\eqref{KillingVectorBdy} in the main text. In order to obtain further
insight into the content of the reduced Killing boundary equations we
analyse the conditions under which it is possible to ensure the
existence of such conformal Killing vector in terms of assumptions on
the conformal boundary and initial data at the corner. To this end, we
mimic the analysis on the spacetime and consider a conformal Killing
equation propagation system intrinsic to the boundary ---see the
equations in \Cref{WaveEquationsZeroQuantities}. This systems
allows the identification of an \emph{obstruction tensor}
$\mathcal{O}_{ab}$, constructed from an intrinsic conformal Killing
vector candidate and the Cotton tensor of the conformal boundary,
whose vanishing ensures the existence of the required intrinsic
conformal Killing vector ---see equation \eqref{Obstruction}. In particular, if the conformal boundary is
conformally flat (as in the case, for example, of the Kerr-anti de
Sitter spacetime) then the obstruction tensor vanishes. The existence
of the conformal Killing vector intrinsic to the conformal boundary is
formulated in \Cref{MainPropBoundary}. Finally, our main result
concerning the existence of Killing vectors in the development of an
initial-boundary value problem for the conformal Einstein equations is
given in \Cref{MainTheorem}

An important property of the analysis described in the previous
paragraphs which follows from working in an unphysical
(i.e. conformally rescaled) spacetime is that the boundary conditions
(both at a spacetime and intrinsic level) 
required for the existence of a Killing vector in the physical
spacetime are conformally invariant. Thus, the analysis is independent
of the conformal representation one is working with.

\medskip
An alternative approach to the analysis of continuous symmetries in
anti-de Sitter-like spacetimes has been started in
\cite{HolSha16,HolSha16b}. In this work, the objective is to encode
the existence of a Killing vector solely through conditions on the
conformal boundary ---in the spirit of the principle of
\emph{holography}. The required analysis, thus, leads to the study of
ill-posed initial value problems for wave equations which require the
use of methods of the \emph{theory of unique continuation}. Their
analysis requires imposing both Dirichlet and Neuman boundary
conditions on the conformal boundary while the discussion in the
present work requires, as already mentioned, only Dirichlet
conditions. The trade off is that our analysis also requires a
solution to the KID equation on a spacelike hypersurface and
compatibility conditions between the Killing initial and boundary
data.

\subsection*{Conventions}
Through out, the term \emph{spacetime} will be used to denote a
4-dimensional Lorentzian manifold which not necessarily satisfies the
Einstein field equations. Moreover,
$(\tilde{\mathcal{M}},\tilde{g}_{ab})$ will denote a vacuum spacetime
satisfying the Einstein equations with anti de Sitter-like
cosmological constant $\lambda$.  The signature of the metric in this
article will be $(-,+,+,+)$. It follows that $\lambda<0$. The
lowercase Latin letters $a,\, b,\, c, \ldots$ are used as abstract
spacetime tensor indices while the indices $i,\,j,\,k,\ldots$ are
abstract indices on the tensor bundle of hypersurfaces of
$\tilde{\mathcal{M}}$. The Greek letters $\mu, \, \nu, \,
\lambda,\ldots$ will be used as spacetime coordinate indices while
$\alpha,\,\beta,\,\gamma,\ldots$ will serve as spatial coordinate
indices.

Our conventions for the curvature are
\[
\nabla_c \nabla_d u^a -\nabla_d \nabla_c u^a = R^a{}_{bcd} u^b.
\]

\section{The metric conformal Einstein field equations}

Throughout all this work we will make use of the Einstein equations on a
conformal setting. Therefore, in this section the properties of 
this representation will be presented.

\medskip
Let $(\tilde{\mathcal{M}},\tilde{g}_{ab})$ a 4-dimensional spacetime 
satisfying the vacuum Einstein field equations
\begin{equation}
\tilde{R}_{ab} = \lambda \tilde{g}_{ab},
\label{EFE}
\end{equation}
where $\tilde{R}_{ab}$ is the Ricci tensor associated to the metric
$\tilde{g}_{ab}$ and $\lambda$ the so-called cosmological constant.
Now, consider a conformal embedding 
consider a spacetime $(\mathcal{M},g_{ab})$ which is related to
$(\tilde{\mathcal{M}},\tilde{g}_{ab})$ via a \emph{conformal embedding}
\[
\tilde{\mathcal{M}}\stackrel{\varphi}{\hookrightarrow}\mathcal{M},
\qquad \tilde{g}_{ab} \stackrel{\varphi}{\mapsto} g_{ab}\equiv
\Xi^2 \big( \varphi^{-1})^*\tilde{g}_{ab}, \qquad
\Xi|_{\varphi(\tilde{\mathcal{M}})}>0.
\]
Slightly abusing of the notation we
write
\begin{equation}
g_{ab} = \Xi^2 \tilde{g}_{ab},
\label{ConformalRescaling}
\end{equation}
where the \emph{conformal factor} $\Xi$ is a non-negative scalar function. The
set of points of $\mathcal{M}$ for which $\Xi$ vanishes will be called the
\emph{conformal boundary}. We use the notation $\mathscr{I}$ to denote the
parts of the conformal boundary which are an hypersurface of $\mathcal{M}$.

\subsection{Basic properties}

In what follows, let $\nabla_a$ denote the Levi-Civita connection of the metric
$g_{ab}$.  Let $R^a{}_{bcd}$, $R_{ab}$, $R$ and $C^a{}_{bcd}$ denote,
respectively, the corresponding Riemann tensor, Ricci tensor, Ricci scalar and
(conformally invariant) Weyl tensor. In a conformal context
it is customary to introduce \emph{Schouten tensor} $L_{ab}$,
defined as
\[
L_{ab} \equiv \frac{1}{2}\bigg(R_{ab} - \frac{1}{6} R g_{ab}\bigg).
\]
Moreover, it is useful to define the following quantities:
\[
s \equiv\frac{1}{4} \nabla^c\nabla_c \Xi + \frac{1}{24}R\Xi, \qquad d^a{}_{bcd} \equiv \Xi^{-1}C^a{}_{bcd},
\]

where the former is the so-called \emph{Friedrich scalar} and the latter is the
\emph{rescaled Weyl tensor}.

In terms of the objects defined above, and under a conformal transformation,
the Einstein equations \eqref{EFE} imply a system of differential equations
known as the \emph{metric vacuum conformal Einstein field equations}, given by:
\begin{subequations}
\begin{eqnarray}
&& \nabla_a \nabla_b \Xi =-\Xi L_{ab} + sg_{ab}, \label{CFE1}\\
&& \nabla_a s = -L_{ac} \nabla^c \Xi, \label{CFE2}\\
&& \nabla_a L_{bc} - \nabla_b L_{ac} = \nabla_e \Xi d^e{}_{cab}, \label{CFE3}\\
&& \nabla_e d^e{}_{cab} =0,  \label{CFE4}\\
&& \lambda  = 6 \Xi s - 3 \nabla_c \Xi \nabla^c \Xi. \label{CFE5}
\end{eqnarray}
\end{subequations}
A detailed derivation of this system for the general case of a non-zero matter
component can be found in \cite{CFEBook}.

\begin{remark}
{\em Expressions \eqref{CFE1}-\eqref{CFE4} are differential equations for the
fields $\Xi$, $s$, $L_{ab}$ and $d^a{}_{bcd}$, while equation \eqref{CFE5} will
be regarded as a constraint. As shown in Lemma 8.1 in \cite{CFEBook}, if
\eqref{CFE1} and \eqref{CFE2} are satisfied, \eqref{CFE5} will automatically do
so as long as it holds at a single point}.
\end{remark}

\medskip
By a solution to the metric conformal Einstein field equations it is understood a collection
\[
(g_{ab}, \Xi, s, L_{ab}, d^a{}_{bcd})
\]
satisfying equations \eqref{CFE1}-\eqref{CFE5}. If $\tilde{g}_{ab}$ is
a solution to the Einstein equations \eqref{EFE} and it is conformally
related to $g_{ab}$, then the latter is a solution to the conformal
Einstein field equations. The converse of this statement is given as
follows:

\begin{proposition}
Let $(g_{ab}, \Xi, s, L_{ab}, d^a{}_{bcd})$ denote a solution to the
metric conformal Einstein field equations \eqref{CFE1}-\eqref{CFE4} such that $\Xi\neq 0$ on an
open set $\mathcal{U}\subset \mathcal{M}$. If, in addition, equation
\eqref{CFE5} is satisfied at a point $p\in \mathcal{U}$, then the
metric 
\[
\tilde{g}_{ab} = \Xi^{-2} g_{ab}
\]
is a solution to the
Einstein field equations \eqref{EFE} on $\mathcal{U}$. 
\end{proposition}
A proof of this proposition is given in \cite{CFEBook} ---see
Proposition 8.1 in that reference.

\medskip
The causal character of the conformal boundary $\mathscr{I}$ is determined by
the sign of the Cosmological constant. As this will be of key importance in the
forthcoming sections, we make this more precise:

\begin{proposition}
Suppose that the Friedrich scalar $s$ is regular on $\mathscr{I}$. Then
$\mathscr{I}$ is a null, spacelike or timelike hypersurface of
$\mathcal{M}$, respectively, depending on whether $\lambda=0$,
$\lambda>0$ or $\lambda<0$. 
\end{proposition}

\begin{proof}
This result follows directly from evaluating equation \eqref{CFE5}
at $\mathscr{I}$ and recalling that $\nabla_a\Xi$ is normal to this hypersurface.
\end{proof}

\subsection{Wave equations for the conformal fields}

In \cite{Pae15} it has been shown how the conformal Einstein field equations
\eqref{CFE1}-\eqref{CFE4} imply a system of \emph{geometric} wave equations for
the components of the fields $(\Xi,\,s,\,L_{ab},\,d^a{}_{bcd})$.
Such system takes the form:

\begin{proposition}
Any solution $(\Xi,\,s,\,L_{ab},\,d^a{}_{bcd})$ to the vacuum conformal Einstein field
equations \eqref{CFE1}-\eqref{CFE4} satisfies the equations
\begin{subequations}
\begin{eqnarray}
&&\square \Xi = 4s -\frac{1}{6}\Xi R, \label{WaveCFE1}\\
&&\square s=  \Xi L_{ab} L^{ab}- \frac{1}{6} s R -  \frac{1}{6} \nabla_{a}R \nabla^{a}\Xi, \label{WaveCFE2}\\
&& \square L_{ab}= 4 L_{a}{}^{c} L_{bc}- g_{ab} L_{cd} L^{cd} - 2 \Xi d_{acbd} L^{cd} 
+ \frac{1}{6} \nabla_{a}\nabla_{b}R, \label{WaveCFE3}\\
&& \square d_{abcd} = 2 \Xi d_{a}{}^{e}{}_{d}{}^{f} d_{becf} 
- 2 \Xi d_{a}{}^{e}{}_{c}{}^{f} d_{bedf} - 2 \Xi d_{ab}{}^{ef} 
d_{cedf} + \frac{1}{2} d_{abcd} R.\label{WaveCFE4}
\end{eqnarray}
\end{subequations}
\end{proposition} 

\section{Killing vectors in the conformal setting}
In this section we briefly review the theory of Killing vectors from a
conformal point of view. Our presentation follows that of \cite{Pae16}.


\subsection{Conformal properties of the Killing vector equation}

We begin by recalling the relation between Killing vectors in
the physical spacetime $(\tilde{\mathcal{M}}, \tilde{g}_{ab})$ 
and conformal Killing vectors in the unphysical
spacetime $(\mathcal{M}, g_{ab})$:

\begin{lemma}
A vector field $\tilde\xi^a$ is a Killing vector field of
$(\tilde{\mathcal{M}},\tilde{g}_{ab})$, that is
\[
\tilde\nabla_a\tilde\xi_b + \tilde\nabla_b \tilde\xi_a =0,
\]
 if and only if its push-forward
$\xi^a\equiv \varphi_* \tilde{\xi}^a$ is a conformal Killing vector field
in $(\mathcal{M},g_{ab})$, i.e.
\begin{equation}
\nabla_a \xi_b + \nabla_b \xi_a = \frac{1}{2}\nabla_c \xi^c g_{ab}
\label{UnphysicalKillingEquation1}
\end{equation}
and, moreover, one has that 
\begin{equation}
\xi^a \nabla_a \Xi = \frac{1}{4}\Xi \nabla_a \xi^a.
\label{UnphysicalKillingEquation2}
\end{equation}
\end{lemma}

The proof of this result can be found in \cite{Pae16}. 

\begin{remark}
{\em In the following we will call equations
  \eqref{UnphysicalKillingEquation1} and
  \eqref{UnphysicalKillingEquation2} the \emph{unphysical Killing
    equations}. Observe that if $g_{ab}$ extends smoothly across
  $\mathscr{I}$, then the unphysical Killing equations are well
  defined  at the conformal boundary. }
\end{remark}
This leads to a natural question about the conditions for the existence of
unphysical Killing vectors. This will be addressed in the remaining of this
section.

\subsection{Necessary conditions}

For convenience set 
\[
\eta \equiv \frac{1}{4}\nabla_a \xi^a. 
\]
Then one has the following result:

\begin{lemma}
\label{WEKillingSpacetime}
Any solution to the unphysical Killing equations satisfies the system
\begin{subequations}
\begin{eqnarray}
&& \square \xi_a + R_a{}^b \xi_b + 2 \nabla_a \eta =0, \label{BoxXi}\\
&& \square \eta + \frac{1}{6}\xi^a \nabla_a R + \frac{1}{3}R\eta =0. \label{BoxEta}
\end{eqnarray}
\end{subequations}
\end{lemma}
The proof of the above result follows by direct computation from
\eqref{UnphysicalKillingEquation1} and \eqref{UnphysicalKillingEquation2}.

\begin{remark}
{\em The wave equations \eqref{BoxXi} and \eqref{BoxEta} are necessary
conditions for a vector $\xi^a$ to be an \emph{unphysical Killing vector}.
However, not every solution to these equations is an unphysical Killing vector.
In this sense, a vector field satisfying \eqref{BoxXi}-\eqref{BoxEta} will be
called a \emph{unphysical Killing vector candidate}. }
\end{remark}

\subsection{The unphysical Killing equation propagation system}
\label{Subsection:UnphysicalKillingEquationPropagationSystem}

The sufficient conditions are now discussed. It will be convenient to define
the following \emph{zero-quantities}:
\begin{eqnarray*}
&& S_{ab} \equiv \nabla_a \xi_b + \nabla_b \xi_a -2 \eta g_{ab}, \\
&& S_{abc} \equiv \nabla_a S_{bc}, \\
&& \phi \equiv \xi^a\nabla_a\Xi -\Xi \eta, \\
&& \psi \equiv \eta s + \xi^a \nabla_a s - \nabla_a \eta \nabla^a \Xi,
  \\
&& B_{ab} \equiv \mathcal{L}_\xi L_{ab} + \nabla_a\nabla_b\eta,
\end{eqnarray*}
with $\mathcal{L}_\xi$ denoting the \emph{Lie derivative} along the direction
of $\xi^a$. Recall that 
\[
\mathcal{L}_\xi L_{ab} = \xi^c \nabla_c L_{ab} 
+ L_{cb}\nabla_a \xi^c + L_{ac}\nabla_b \xi^c. 
\]
In terms of these quantities, a lengthy computation leads to the following
result proved in \cite{Pae16}:

\begin{lemma}
\label{WEZQ}
Let $\xi^a$ and $\eta$ be a pair of fields satisfying equations
\eqref{BoxXi}-\eqref{BoxEta}. One then has that the tensor fields 
\[
S_{ab}, \quad S_{abc}, \quad \phi, \quad, \psi, \quad  B_{ab},
\]
satisfy a closed system of homogeneous wave equations. Schematically
one has that
\begin{subequations}
\begin{eqnarray}
&& \square S_{ab} =  H_{ab} (S,B), \label{KillingEvolution1}\\
&& \square S_{abc}=  H_{abc} (S,B,\nabla S, \nabla B), \label{KillingEvolution2}\\
&& \square \phi = H(\phi,\psi, S), \label{KillingEvolution3}\\
&& \square \psi = K(\phi,S,B,\psi,\nabla\phi), \label{KillingEvolution4}\\
&& \square B_{ab} = K_{ab}(S,B,\nabla S, \nabla B, \nabla^2S). \label{KillingEvolution5}
\end{eqnarray}
\end{subequations}
\end{lemma}

\begin{remark}
{\em In what follows the system consisting of equations
  \eqref{BoxXi}-\eqref{BoxEta} together with
  \eqref{KillingEvolution1}-\eqref{KillingEvolution5} will be called
  the \emph{unphysical Killing equation propagation system}. }
\end{remark}

The homogeneity of the unphysical Killing equation evolution system
\eqref{KillingEvolution1}-\eqref{KillingEvolution5} together with the theory of
initial-boundary value problems for systems of wave equations (see
e.g. \cite{CheWah83,DafHru85}) suggests to
consider a Dirichlet problem to ensure the existence of a solution to the
unphysical Killing vector equations. Let $\mathcal{S}_\star$ be an initial spacelike
hypersurface. The conditions for the problem are:
\begin{enumerate}[(i)]
\item Initial data 
\begin{subequations}
\begin{eqnarray}
& S_{ab}=0, \quad S_{abc}=0, \quad \phi=0, \quad \psi=0, \quad
  B_{ab}=0, &  \label{ConformalKID1} \\
& \nabla_e S_{ab}=0, \quad \nabla_e S_{abc}=0, \quad \nabla_e \phi=0,
  \quad \nabla_e \psi=0, \quad
  \nabla_e B_{ab}=0, \qquad \mbox{on} \quad \mathcal{S}_\star;&
\label{ConformalKID2}
\end{eqnarray}
\end{subequations}
\item(Dirichlet) boundary data
\begin{equation}
 S_{ab}=0, \quad S_{abc}=0, \quad \phi=0, \quad \psi=0, \quad
  B_{ab}=0, \qquad \mbox{on} \quad \mathscr{I}.
\label{VanishingBoundaryData}
\end{equation}
\end{enumerate}

If the above conditions are satisfied, the homogeneity of the wave equations
\eqref{KillingEvolution1}-\eqref{KillingEvolution5} guarantees that the only
solution of the system is the trivial one. This means, therefore, that the
solution to equations \eqref{BoxXi}-\eqref{BoxEta} will actually be an
unphysical Killing vector. 

\begin{remark}
\em{ Strictly speaking, the initial conditions require only the vanishing of
the zero-quantities and of their normal derivatives to the initial
hypersurface. If these conditions hold then the full covariant derivative of
the zero-quantities vanish initially and conversely. }
\end{remark}


\section{The conformal constraint equations}

In order to investigate conditions for the Dirichlet problem, we
recall that the conformal
Einstein equations impose some restrictions on the conformal boundary. In this
context a $3+1$ decomposition arises as a natural approach to the problem.

\subsection{The 3 + 1 decomposition of the conformal field equations}

Let {$\mathcal{K} \subset \mathcal{M}$ be a 3-dimensional  hypersurface with
normal vector $n_a$. The hypersurface $\mathcal{K}$ is endowed with a
metric $k_{ab}$
\footnote{In this work, intrinsic 3-dimensional objects will be regarded as
living on the spacetime, so they will be denoted using Latin indices
taken from the first part of the alphabet.} related
to the spacetime one via:
\[
k_{ab} = g_{ab} - \epsilon n_a n_b,
\]
where $\epsilon \equiv n_a n^a$ take either the value $1$ if
$\mathcal{K}$ is timelike or $-1$ if it is spacelike. The nilpotent operator
$k_a{}^b$ effectively projects spacetime objects into $\mathcal{K}$.
Moreover, it induces a decomposition of the covariant derivative via
the relation
\[
\nabla_a = k_a{}^b \nabla_b + \epsilon n_a n^b\nabla_b \equiv
D_a + \epsilon n_a D.
\]
Here, $D_a$ is the covariant derivative intrinsic to $\mathcal{K}$ which
satisfies the metric compatibility condition $D_a k_{bc} = 0$, and $D$
corresponds to the derivative in the normal direction. Additionally, the
intrinsic curvature associated to $\mathcal{K}$, denoted by $K_{ab}$, can be
conveniently expressed in terms of the acceleration $a_b \equiv n^c\nabla_c n_b$ as
\[
\nabla_a n_b = K_{ab} + n_a a_b.
\]

The fields appearing in the conformal Einstein field equations can be
naturally decomposed using the projector $k_a{}^b$. Relevant for
the subsequent work, let
\[
\Sigma, \quad s, \quad k_{ab}, \quad \theta_a, \quad \theta_{ab}, \quad d_{ab},
\quad d_{abc}
\]
denote, respectively, the pull-backs of 
\begin{eqnarray*}
&n^a \nabla_a \Xi, \quad s, \quad g_{ab}, \quad n^c k_a{}^d L_{cd},
  \quad k_a{}^ck_b{}^d L_{cd}, \quad 
n^b n^d k_e{}^a k_f{}^cd_{abcd}, \quad n^b k_e{}^a k_f{}^c
k_g{}^d d_{abcd}
\end{eqnarray*}
to $\mathcal{K}$. 

\begin{remark}
{\em The fields $d_{ab}$ and $d_{abc}$ represent, respectively, the
\emph{electric} and \emph{magnetic parts} of the rescaled Weyl tensor
$d_{abcd}$ with respect to the normal $n_a$. The following properties can be
verified:
\begin{eqnarray*}
& d_a{}^a=0, \quad d_{ab}=d_{ba}, \quad d_{abc}=-d_{acb}, \quad d_{[abc]}=0.
\end{eqnarray*}}
\end{remark}

\subsection{The conformal constraint equations} 

When the Einstein field equations \eqref{CFE1}-\eqref{CFE5} are projected into
a hypersurface $\mathcal{K}$ via $k_a{}^b$, the result is a system known as
the \emph{conformal constraint equations}. In terms of the quantities defined
above, a long computation results in the system
\begin{subequations}
\begin{eqnarray}
&& D_i D_j \Omega = -\epsilon \Sigma K_{ij} -\Omega L_{ij} + s
k_{ij}, \label{ConformalConstraint1} \\
&& D_i \Sigma = K_i{}^k D_k \Omega -\Omega L_i, \label{ConformalConstraint2} \\
&& D_i s = -\epsilon L_i \Sigma - L_{ik} D^k \Omega, \label{ConformalConstraint3}\\
 && D_i L_{jk} -D_j L_{ik} = -\epsilon \Sigma d_{kij}
+ D^l \Omega d_{lkij} -\epsilon (K_{ik} L_j
- K_{jk} L_i), \label{ConformalConstraint4}\\
 && D_i L_j - D_j L_i = D^l \Omega d_{lij} +
K_i{}^k L_{jk} - K_j{}^k L_{ik}, \label{ConformalConstraint5}\\
&& D^k d_{kij} =\epsilon \big(K^k{}_i d_{jk}
   -K^k{}_j d_{ik}\big), \label{ConformalConstraint6}\\
&& D^i d_{ij}= K^{ik} d_{ijk}, \label{ConformalConstraint7}\\
&& \lambda = 6 \Omega s - 3\epsilon \Sigma^2 - 3 D_k \Omega D^k\Omega.
\label{ConformalConstraint8}
\end{eqnarray}
\end{subequations}
Additionally, these are supplemented by the conformal versions of the
Codazzi-Mainardi and Gauss-Codazzi equations. These are, respectively:
\begin{subequations}
\begin{eqnarray}
&& D_j K_{ki} - D_k K_{ji} = \Omega d_{ijk} + k_{ij} L_k -
 k_{ik}L_j, \label{ConformalConstraint9} \\
&& l_{ij} = -\epsilon\Omega d_{ij} + L_{ij} + \epsilon \bigg( K \big(
 K_{ij} -\displaystyle\frac{1}{4} K k_{ij}\big) - K_{ki}
  K_j{}^k + \displaystyle\frac{1}{4}  K_{kl} K^{kl}k_{\bmi\bmj}\bigg).
  \label{ConformalConstraint10}
\end{eqnarray}
\end{subequations}
Here, $l_{ab}$ is the 3-dimensional Schouten tensor, given in terms
of the associated Ricci tensor and scalar $r_{ab}$ and $r$, respectively, by
\[
l_{ab} \equiv r_{ab} - \frac14 rk_{ab}.
\]
A detailed derivation of this equations, as well as a discussion about some of
their properties, can be found in \cite{CFEBook}.  In the following it will be shown that, under
a gauge choice, this system enables us to analyse the conformal boundary in a
simpler way.

\subsubsection{The conformal constraints on $\mathscr{I}$}

Hereafter, $\simeq$ will denote equality at the conformal boundary $\mathscr{I}$.  When the
constraints \eqref{ConformalConstraint1}-\eqref{ConformalConstraint8}, along
with \eqref{ConformalConstraint9} and \eqref{ConformalConstraint10} are
evaluated on  $\mathscr{I}$ ---for which $\epsilon = 1$--- they take a particularly
simple form as, by definition, the conformal factor identically vanishes. It follows that
the constraints on $\mathscr{I}$ are:
\begin{subequations}
\begin{eqnarray}
&& s \ell_{ab} \simeq \Sigma K_{ab}, \label{CCCB1}\\
&& D_a \Sigma \simeq 0, \label{CCCB2}\\
&& D_a s \simeq -\Sigma \theta_a, \label{CCCB3}\\
&& D_a \theta_{bc} - D_b \theta_{ac} \simeq -\Sigma d_{cab} + (K_{bc}\theta_a - K_{ac}\theta_b),
\label{CCCB4}\\
&& D_a \theta_b - D_b \theta_a \simeq K_a{}^c \theta_{bc} - K_b{}^c \theta_{ac}, \label{CCCB5}\\
&& D^c d_{cab} \simeq K^a{}_b d_{ac} - K^c{}_a d_{bc}, \label{CCCB6}\\
&& D^a d_{ab}\simeq K^{ac} d_{abc}, \label{CCCB7}\\
&& \lambda \simeq -3 \Sigma^2, \label{CCCB8}\\
&& D_b K_{ac} - D_c K_{ab}\simeq \ell_{ab} \theta_c - \ell_{ac} \theta_b, \label{CCCB9}\\
&& l_{ab} \simeq \theta_{ab} + K\big( K_{ab} - \frac{1}{4} K \ell_{ab} \big) -
   K_{ac}K_b{}^c + \frac{1}{4}K_{cd}K^{cd} \ell_{ab}. \label{CCCB10}
\end{eqnarray}
\end{subequations}
In \cite{Fri95}, an approach to find a solution of the above system has been given.
The main characteristic of this method resides in regarding $s$ as a gauge
quantity. Such result can be enunciated as follows:

\begin{proposition}
\label{ConstSolBdy}
Given a 3-dimensional Lorentzian metric $\ell_{ab}$, a smooth function
$\varkappa$ and a symmetric field $d_{ab}$ satisfying $\ell^{ab}d_{ab} =0 $ and
$D^ad_{ab}=0$, then the following fields are a solution to the conformal
constraint equations \eqref{CCCB1}-\eqref{CCCB10} on $\mathscr{I}$:
\begin{subequations}
\begin{eqnarray}
&& \Sigma \simeq \sqrt{\frac{|\lambda|}{3}}, \label{ConstSolBdy1} \\
&& s \simeq \varkappa\Sigma, \label{ConstSolBdy2} \\
&& K_{ab} \simeq \varkappa\ell_{ab}, \label{ConstSolBdy3} \\
&& \theta_a \simeq -D_a\varkappa, \label{ConstSolBdy4} \\
&& \theta_{ab} \simeq l_{ab} - \frac12\varkappa^2 \ell_{ab}, \label{ConstSolBdy5} \\
&& d_{abc} \simeq 
-\Sigma^{-1}y_{abc}, \label{ConstSolBdy6}
\end{eqnarray}
\end{subequations}
where $y_{abc} \equiv D_b l_{ca} - D_c l_{ba}$ is the Cotton tensor of the metric
$\bm\ell$.
\end{proposition}


\section{Decomposition of the zero-quantities}

The  $3+1$ decomposition described in the previous section is also
key to study the zero-quantities associated to the Killing
vector equation on a given hypersurface
$\mathcal{K}$. In this respect, let define the following relevant quantities:

\[
\zeta_a, \quad \zeta, \quad
\mathcal{S}_{ab}, \quad \mathcal{S}_a, \quad \mathcal{S}, \quad
\mathcal{S}_{abc}, \quad \mathcal{B}_{ab} \quad \mathcal{B}_a, \quad \mathcal{B}
\]
as the respective the pull-backs of the following projections of the Killing
vector candidate $\xi_a$ and the zero-quantities into
$\mathcal{K}$:
\begin{eqnarray*}
& k_a{}^b \xi_b, \quad n^a\xi_a, \quad 
k_a{}^ck_b{}^dS_{cd}, \quad n^ck_a{}^bS_{bc}, \quad n^an^b S_{ab}, \quad
k_a{}^d k_b{}^e k_c{}^f S_{def}, & \\ 
& k_a{}^ck_b{}^d B_{cd}, \quad
n^ck_a{}^bB_{bc}, \quad n^an^bB_{ab}. &
\end{eqnarray*}
In the next subsection, the vanishing of the zero-quantities
on $\mathcal{S}_\star$ and $\mathscr{I}$ will be analysed using these objects.

\begin{remark}
{\em As mentioned in Section
\ref{Subsection:UnphysicalKillingEquationPropagationSystem}, the initial
data for the wave equations
\eqref{KillingEvolution1}-\eqref{KillingEvolution5} requieres the
vanishing of not only the zero-quantities on the initial hypersurface
but also the vanishing of their first order covariant derivatives. Given that we can decompose
$\nabla_a$ in terms of intrinsic and normal operators, then
if the zero-quantities vanish initially so will all their intrinsic
derivatives. Thus, the subsequent analysis only needs to consider normal derivatives.}
\end{remark}

\subsection{Decomposition of $\phi$ and $\psi$}

From their definitions, a straightforward decomposition of the zero-quantities
$\phi$,  $\psi$, and their normal derivatives, leads to the following
expressions:
\begin{subequations}
\begin{eqnarray}
&& \phi  = \zeta^a D_a\Xi + \epsilon\zeta\Sigma - \eta\Xi, \label{DecompPhi} \\
&& n^a\nabla_a \phi  = -\eta\Sigma - \Xi D\eta + 
D\zeta^a D_a\Xi + \zeta^a(D_a\Sigma - K_a{}^c D_c\Xi)
 + \epsilon(\zeta D\Sigma + \Sigma D\zeta), \label{DecompDerPhi}
\end{eqnarray}
\end{subequations}
and
\begin{subequations}
\begin{eqnarray}
&& \psi  = \eta s + \zeta^a D_a s + \epsilon\zeta Ds - D_a\eta D^a\Xi -\epsilon\Sigma D\eta, \label{DecompPsi} \\
&& n^a\nabla_a\psi   = \eta Ds + sD\eta + D\zeta^aD_a s + \zeta^a(D_aDs - K_a{}^bD_bs)
- D^a\eta(D_a\Sigma - K_a{}^b D_b\Xi) \nonumber \\
&& \hskip+1.5cm - D_a\Xi(D^aD\eta - K^a{}_bD^b\eta) +\epsilon(\zeta D^2s + D\zeta Ds - D\Sigma D\eta -\Sigma D^2\eta).
\label{DecompDerPsi}
\end{eqnarray}
\end{subequations}

\subsection{Decomposition of $S_{ab}$ and  $B_{ab}$ and
 their derivatives}

Before performing a decomposition of the remaining
zero-quantities some observations can be made about the redundancy of some of
their components. For this task their explicit decompositions will not be required but
expressions will be given in terms of functions which are homogeneous in some
zero-quantities and their derivatives; this will prove to be useful when
imposing the vanishing initial-boundary data.

\begin{lemma}
\label{CondDerSab}
Let $\mathcal{K} \subset \mathcal{M}$ be either a timelike or spacelike
hypersurface.  Assume that $S_{ab}, \ D\mathcal{S}_{ab}, \ \mathcal{B}_{ab}$
and $D\mathcal{B}_{ab}$ are known on $\mathcal{K}$. Then, the remaining
components of the zero--quantities and their first-order derivatives can be
computed on $\mathcal{K}$.
\end{lemma}

\begin{proof}
In the following, for ease of presentation, let $f$ denote a generic
homogeneous function of its arguments which may change from line to line.
As pointed out in \cite{Pae16}, equation \eqref{BoxXi} implies the identity
\begin{equation}
\nabla_a S_b{}^a - \frac12\nabla_b S_a{}^a = 0.
\label{PaetzIdentitySab}
\end{equation}
Expressing $S_{ab}$ in terms of its components, a short calculation yields
\begin{equation}
\epsilon D\mathcal{S}_b + \frac12 n_bD\mathcal{S} = 
f(\mathcal{S}_{ab}, \ D_a \mathcal{S}_{bc}, \ D\mathcal{S}_{ab}).
\label{PaetzIdentitySabDecomposed}
\end{equation}
Multiplying this equation by $k_a{}^b$, an equation for $D\mathcal{S}_a$ is
obtained.  Similarly, multiplying equation \eqref{PaetzIdentitySabDecomposed}
by $n^b$ we obtain an analogous expression for $D\mathcal{S}$.
Then, all the components of $DS_{ab}$ can be computed on
$\mathcal{K}$ and, in consequence, $\nabla_a S_{bc}$ is known. This
determines $S_{abc}$ on the hypersurface.

In order to analyse the fields derived from $B_{ab}$, consider equation
\eqref{KillingEvolution1} which can be written in a more explicit way as:
\begin{equation}
D^2S_{ab} = -4\epsilon B_{ab} + f(S_{ab},  \ \nabla_cS_{ab}, \ D_cD_d S_{ab}). \label{Obs2WE}
\end{equation}
As it is assumed that $\mathcal{B}_{ab}$ is known on $\mathcal{K}$, then one
can solve for $D^2\mathcal{S}_{ab}$ from this last equation; in particular,
$D^2\mathcal{S}_a{}^a$ can be computed. On the other hand, applying $\nabla_c$
to \eqref{PaetzIdentitySab}, a lengthy but direct decomposition leads to the
following two relations:
\begin{subequations}
\begin{eqnarray}
&& D^2\mathcal{S}_a = f(S_{ab}, \ \nabla_cS_{ab}), \label{Obs2Paetz1} \\
&& \epsilon D^2\mathcal{S} = D^2\mathcal{S}_a{}^a + f(S_{ab}, \ \nabla_cS_{ab}). \label{Obs2Paetz2}
\end{eqnarray}
\end{subequations}
From here we observe that their right-hand sides are either
known or computable on $\mathcal{K}$ so the components $D^2\mathcal{S}_a$ and
$D^2\mathcal{S}$ are determined. Thus, \eqref{Obs2WE} implies that the
components $\mathcal{B}_a$ and $\mathcal{B}$ can be computed.

\smallskip
Regarding the normal derivatives of $B_{ab}$, we make use of the 
identity
\[
\nabla_a B_{b}{}^a - \tfrac12\nabla_{b}B_a{}^a = S_{cd}(\nabla^c L_b{}^d - \tfrac12\nabla_bL^{cd}),
\]
whose validity is guaranteed by equations \eqref{BoxXi} and \eqref{BoxEta}
---see \cite{Pae16}.  Observe that its left hand side has the same form as
equation \eqref{PaetzIdentitySab}, while its right hand side is homogeneous on
$S_{ab}$ ---which is already known. Then we conclude that $D\mathcal{B}_a$ and
$D\mathcal{B}$ are computable.

\smallskip
Finally, the normal derivative of $S_{abc}$ can be analysed from its
definition.  Commuting derivatives, a short calculation yields:
\[
DS_{abc} = D_a(DS_{bc}) + \epsilon n_a D^2S_{bc} + f(S_{ab}, \ \nabla_c S_{ab}).
\]
Since it has been proved that all the terms are either computable or part of
the given data on $\mathcal{K}$, the proof is complete.

\end{proof}

\begin{remark}
\label{TimelikeAssumptions}

\em{\Cref{CondDerSab} is valid either for a spacelike or timelike hypersurface,
but given that it assumes certain normal derivatives, it is naturally adapted
to a spacelike hypersurface where first-order derivatives are
assumed as part of the initial data. If $\mathcal{K}$ is timelike and Dirichlet
conditions are assumed, then $D\mathcal{S}_{ab}$ plays the role of the only
necessary component of $S_{abc}$, while $D\mathcal{B}_{ab}$ is not required.}

\end{remark}

In view of the previous result,  the explicit form of the
remaining independent data under a decomposition on $\mathcal{K}$ is
given by:
\begin{subequations}
\begin{eqnarray}
&& \mathcal{S}_{ab}  = D_a \zeta_b + D_b \zeta_a + 2\epsilon\zeta K_{ab} - 2\eta k_{ab}, \label{DecompSab} \\
&& \mathcal{S}_{a}  = D\zeta_a + D_a\zeta - \zeta^bK_{ab}, \label{DecompSa}\\
&& \mathcal{S}  = 2D\zeta -2\epsilon\eta, \label{DecompS} \\
&& D\mathcal{S}_{ab} = 2D_{(a} D\zeta_{b)} - 2K_{(a}{}^c D_{|c|}\zeta_{b)}
+ 2\zeta^c D_cK_{ab} - 2\zeta^c D_{(a}K_{b)c}
+ 2\epsilon\zeta DK_{ab} \nonumber \\
&& \hskip1.3cm +2\epsilon K_{ab}D\zeta - 2k_{ab} D\eta, \label{NormalDerSab} \\
&& \mathcal{B}_{ab}  = \zeta^{c} D_{c}\theta_{ab} + 2\theta_{c(a} D_{b)}\zeta^{c} 
+ 2\epsilon \zeta K_{(a}{}^{c} \theta_{b)c} + 2\epsilon \theta_{(a} D_{b)}\zeta 
+ \epsilon \zeta D\theta_{ab} + D_a D_b \eta, \label{DecompBab} \\
&& D\mathcal{B}_{ab} = D_c\theta_{ab}D\zeta^c + 
K_c{}^eD_e\theta_{ab} + D_cD\theta_{ab} + 2n^e\theta_{(a}{}^d R_{b)dec} +
2K_{(a}{}^c\theta_{b)c}D\zeta \nonumber \\
&& \hskip+1.3cm +  2\zeta K_{(a}{}^cD\theta_{b)c}
+ \zeta\theta_{c(a}DK_{b)}{}^c + 2D_{(b}\zeta^cD\theta_{a)c} + 2\theta_{c(a}(D_b D\zeta^c - K_{b)}{}^eD_e\zeta^c \nonumber \\
&& \hskip+1.3cm + n^eR_{b)ed}{}^c\zeta^d) + 2D\theta_{(a}D_{b)}\zeta + 2\theta_{(a}D_{b)}D\zeta
- 2\theta_{(a}K_{b)}{}^cD_c\zeta + 2\zeta D^2\theta_{ab} \nonumber \\
&& \hskip+1.3cm + 2D\zeta D\theta_{ab} + D_aD_bD\eta - 2K_{(a}{}^eD_{b)}D_e\eta -D_e\eta D_aK_b{}^e
 - n^eR_{aeb}{}^dD_b\eta
. \label{DecompDBab}
\end{eqnarray}
\end{subequations}

\section{Boundary analysis}

The aim of this section is to discuss the explicit requirements a
well-posed initial-boundary problem with vanishing Dirichlet data
impose on the conformal Killing vector candidate and the related
quantities. As a result of this analysis it will be shown that some
components cannot be freely chosen either on 
$\mathscr{I}$.

%
%
%
%

\subsection{Zero-quantities on $\mathscr{I}$}

In this subsection we study the decomposition for the zero-quantities
associated to the Dirichlet boundary conditions for the Killing vector equation
evolution system. As mentioned in \Cref{TimelikeAssumptions}, the independent
data on $\mathscr{I}$ are given by $\phi, \ \psi$, $\mathcal{S}_{ab}, \
D\mathcal{S}_{ab}$ and $\mathcal{B}_{ab}$.  Evaluating equations
\eqref{DecompPhi}, \eqref{DecompPsi} and \eqref{DecompSab}-\eqref{DecompBab} on
$\mathscr{I}$ one obtains


\begin{subequations}
\begin{eqnarray} 
&& \phi \simeq \Sigma\zeta, \label{PhiBdy} \\
&& \psi \simeq \eta s + \zeta^a D_a s + \zeta Ds - \Sigma D\eta, \label{PsiBdy} \\
&& \mathcal{S}_{ab} \simeq D_{a}\zeta_{b} + D_{b}\zeta_{a} + 2\varkappa\zeta \ell_{ab} - 2 \eta \ell_{ab}, \label{SabBdy} \\
&& \mathcal{S}_a \simeq D_a\zeta + D \zeta_a - \varkappa\zeta_a, \label{SaBdy}\\
&& \mathcal{S} \simeq 2 D\zeta - 2\eta \label{SBdy}, \\
&& D\mathcal{S}_{ab} \simeq 2D_{(a} D\zeta_{b)} - 2\varkappa D_{(a}\zeta_{b)}
+ 2\ell_{ab}\zeta^c D_c\varkappa - 2\zeta_{(a} D_{b)}\varkappa + 2\zeta DK_{ab} \nonumber \\
&& \hskip1.3cm +2\varkappa\ell_{ab}D\zeta - 2\ell_{ab} D\eta, \label{DSabBdy} \\
&& \mathcal{B}_{ab} \simeq \zeta^{c} D_{c} l_{ab} + 2 l_{c(a} D_{b)}\zeta^{c} 
+ 2\varkappa\zeta l_{ab} - 2 D_{(a}\varkappa D_{b)}\zeta 
+ \zeta D\theta_{ab} + D_a D_b \eta. \label{BabBdy}
\end{eqnarray}
\end{subequations}

Imposing Dirichlet vanishing data on $\mathscr{I}$, equations
\eqref{PhiBdy}-\eqref{BabBdy} provide a number of conditions for the
fields and their derivatives on the conformal boundary. Using the definition of
$\eta$ and the result of \Cref{ConstSolBdy} it follows that the set of
independent conditions is given by:

\begin{subequations}
\begin{eqnarray}
&& \zeta \simeq 0, \label{EtaBdy} \\ 
&& D\zeta_a \simeq \varkappa\zeta_a, \label{DerZetaBdy} \\
&& D_a\zeta_b + D_b\zeta_a \simeq 2\eta\ell_{ab}, \label{KillingVectorBdy} \\
&& D\eta \simeq \eta\varkappa + \zeta^c D_c\varkappa, \label{DerEtaBdy} \\
&& \mathcal{L}_\zeta l_{ab} + D_a D_b \eta \simeq 0. \label{DerDerEtaBdy}
\end{eqnarray}
\end{subequations}
Conversely, it is straightforward to check that equations
\eqref{EtaBdy}-\eqref{DerDerEtaBdy} are sufficient to guarantee the vanishing
of the equations \eqref{PhiBdy}-\eqref{BabBdy}. The above discussion leads to
the following proposition:

\begin{proposition}
\label{VanishingZQ}
Let $(\mathcal{M}, g_{ab})$ be a conformal extension of an anti-de Sitter
spacetime $(\mathcal{\tilde{M}}, \tilde{g}_{ab})$ with timelike conformal
boundary $\mathscr{I}$.  Let $\xi^a$ be a conformal Killing vector field
candidate and $\phi, \ \psi, \ S_{ab}, \ B_{ab}$ and $S_{abc}$ be the
corresponding zero-quantities. Then, the zero-quantities in
equations \eqref{PhiBdy}-\eqref{BabBdy} vanish on
$\mathscr{I}$ if and only if the components $\zeta_a, \
\zeta$ and $\eta$ satisfy the conditions \eqref{EtaBdy}-\eqref{DerDerEtaBdy}.
\end{proposition}

\begin{remark}
\em{Equations \eqref{EtaBdy}-\eqref{DerDerEtaBdy} will be called the \emph{Killing
boundary data}. They acquire a simpler form if one makes use of a gauge for which $\varkappa=0$.}
\end{remark}

\subsection{Existence of the intrinsic conformal Killing vector}

As stated in \Cref{VanishingZQ}, one of the necessary conditions under which
the set of zero-quantities vanish on $\mathscr{I}$ is given by
\eqref{KillingVectorBdy}  ---i.e. the transversal component $\zeta_a$
of the conformal Killing vector candidate has to be a
conformal Killing vector with respect to the connection $D_a$. In order to
guarantee the existence of a solution to this equation we consider an
initial value problem on $\mathscr{I}$. Following the model of the
spacetime problem, we construct a
suitable wave equation for $\zeta_a$. More precisely, one has the following

\begin{lemma} 
\label{WaveEquationsKillingVector}
Let $\zeta_a$ and $\eta$ a pair of fields satisfying the conformal Killing
equation \eqref{KillingVectorBdy} and \eqref{DerDerEtaBdy} on
$\mathscr{I}$. Then, it follows that 
\begin{subequations}
\begin{eqnarray}
&& \Delta\zeta_a \simeq - r_a {}^b \zeta_b - D_a\eta, \label{WEBdy1} \\
&& \Delta\eta \simeq
- \tfrac{1}{2} \eta r -  \tfrac{1}{4} \zeta^{b} D_{b}r, \label{WEBdy2}
\end{eqnarray}
\end{subequations}
where $\Delta \equiv \ell^{ab}D_aD_b$ is the D'Alambertian operator of the
metric $\ell_{ab}$.
\end{lemma}

\begin{proof}
The result is readily obtained by applying $D^a$ to \eqref{KillingVectorBdy} and
taking the trace of \eqref{DerDerEtaBdy}.
\end{proof}

\begin{remark} \label{Corner} \em{Given that this system of wave equations
propagates $\eta$ and $\zeta_a$ along the conformal boundary, it must be
provided with initial data at the corner $\partial\mathcal{S} = \mathcal{S}_\star
\cap \mathscr{I}$, where $\mathcal{S}_\star \subset \mathcal{M}$ is some initial spacelike
hypersurface.} \end{remark}

To prove that a solution to these wave equations also solves the conformal
Killing equation on the boundary, a suitable system of wave equations for the
corresponding 3-dimensional zero-quantities has to be constructed.  The desired
relations are contained in the following lemma:

\begin{lemma}
\label{WaveEquationsZeroQuantities}
Let $\mathcal{S}_{ab}, \ \mathcal{S}_{abc}$ and $\mathcal{B}_{ab}$ be the
projections of the zero-quantities $S_{ab}, \ S_{abc}$ and $B_{ab}$ into
$\mathscr{I}$, respectively. Assume that there exist fields $\zeta_a$ and
$\eta$ on $\mathscr{I}$ satisfying the wave equations \eqref{WEBdy1} and
\eqref{WEBdy2} in \Cref{WaveEquationsKillingVector}. Then, one has that
\begin{subequations}
\begin{eqnarray*}
&& \Delta\mathcal{S}_{ab} \simeq l_{b}{}^{c} \mathcal{S}_{ac} + l_{a}{}^{c} \mathcal{S}_{bc} 
- 2 r_{acbd} \mathcal{S}^{cd} - 2 \mathcal{B}_{ab} \\
&& \Delta \mathcal{S}_{eab} \simeq
r_{e}{}^{c} \mathcal{S}_{cab} - 2 r_{bdec} \mathcal{S}^{c}{}_{a}{}^{d}
- 2 r_{adec} \mathcal{S}^{c}{}_{b}{}^{d} -  \tfrac{1}{2} r \mathcal{S}_{eab} + r_{b}{}^{c} \mathcal{S}_{eac} 
+ r_{a}{}^{c} \mathcal{S}_{ebc} - 2 r_{acbd} \mathcal{S}_{e}{}^{cd} \\
&& \hspace{1.4cm} + \mathcal{S}_{b}{}^{c} D_{a}r_{ec}
 + \mathcal{S}_{a}{}^{c} D_{b}r_{ec} -  \mathcal{S}_{b}{}^{c} D_{c}r_{ae} -  \mathcal{S}_{a}{}^{c} D_{c}r_{be}
+ \mathcal{S}_{b}{}^{c} D_{e}r_{ac} + \mathcal{S}_{a}{}^{c} D_{e}r_{bc} \\
&& \hspace{1.4cm} - \tfrac{1}{2} \mathcal{S}_{ab} D_{e}r - 2 \mathcal{S}^{cd} D_{e}r_{acbd} - 2 D_{e}\mathcal{B}_{ab} \\
&& \Delta \mathcal{B}_{ab} \simeq \mathcal{O}_{ab}
+ f(\mathcal{B}_{ab}, \ \mathcal{S}_{ab}, \ \mathcal{S}_{abc}, \ D_c\mathcal{S}_{ab}, \ 
D_d \mathcal{S}_{abc}),
\end{eqnarray*}
\end{subequations}
where 
\begin{equation}
\label{Obstruction}
\mathcal{O}_{ab} \equiv \mathcal{L}_{\zeta}D_c y_a{}^c{}_b
+ 2\eta D_c y_a{}^c{}_b + 2 D_c\eta y_{(a}{}^c{}_{b)}
\end{equation}
and $f$ is a homogeneous function of its arguments.
\end{lemma}

\begin{proof}
The wave equations for $\mathcal{S}_{ab}$ and $\mathcal{S}_{abc}$ are obtained
by direct calculation. For the zero--quantity $\mathcal{B}_{ab}$ we have the
two following identities:
\begin{subequations}
\begin{eqnarray*}
&& D_{a}\mathcal{B}_{b}{}^{a} \simeq \tfrac{1}{2} l^{ac} \mathcal{S}_{bac} + \mathcal{S}^{ac} D_{c}l_{ba}, \\
&& D_{a}\mathcal{B}_{bc} \simeq \tfrac{1}{2} \theta_{c}{}^{d} \mathcal{S}_{abd} + 
\tfrac{1}{2} \theta_{c}{}^{d} \mathcal{S}_{bad} -  \tfrac{1}{2} \theta_{a}{}^{d} \mathcal{S}_{bcd}
- \tfrac{1}{2} \theta_{a}{}^{d} \mathcal{S}_{cbd} -  \tfrac{1}{2} \theta_{c}{}^{d} \mathcal{S}_{dab}
+ \tfrac{1}{2} \theta_{a}{}^{d} \mathcal{S}_{dbc} 
+ d_{bcd} \Sigma D_{a}\zeta^{d} \\
&& \hspace{1.4cm} - d_{dac} \Sigma D_{b}\zeta^{d} - d_{bad} \Sigma D_{c}\zeta^{d} + D_{c}\mathcal{B}_{ab} 
- \zeta^{d} \Sigma D_{d}d_{bac}.
\end{eqnarray*}
\end{subequations}
Applying the $D^a$ operator to the latter expression and then using the former
one, as well as using the Bianchi identities, one has that:
\begin{subequations}
\begin{eqnarray*}
&& \Delta \mathcal{B}_{ab} \simeq 
2 \eta \Sigma D_{c}d_{ab}{}^{c} + \Sigma \zeta^{e} D_{e}D_{c}d_{ab}{}^{c} -
 \Sigma d_{cab} D^{c}\eta + 2\Sigma d_{bac} D^{c}\eta +
 \Sigma D_a \zeta^{c} D_{e}d_{cb}{}^e \\
&& \hspace{1.3cm} + \Sigma D_{c}\zeta^b D_e d_{ac}{}^e
+ f(\mathcal{B}_{ab}, \ \mathcal{S}_{ab}, \ \mathcal{S}_{abc}, \
D_c\mathcal{S}_{ab}, \ D_d \mathcal{S}_{abc}) \\
&& \hspace{0.9cm} \simeq - \mathcal{L}_{\zeta}D_c y_{ab}{}^c - 2 \eta D_{c}y_{ab}{}^{c}
 - 2  D^c\eta y_{abc} -  D^c\eta y_{cab} 
+ f(\mathcal{B}_{ab}, \ \mathcal{S}_{ab}, \ \mathcal{S}_{abc}, \ D_c\mathcal{S}_{ab}, \ 
D_d \mathcal{S}_{abc}) \\
&& \hspace{0.9cm} \simeq \mathcal{L}_{\zeta}D_c y_a{}^c{}_b
+ 2\eta D_c y_a{}^c{}_b + D_c\eta(y_a{}^c{}_b + y_b{}^c{}_a)
+ f(\mathcal{B}_{ab}, \ \mathcal{S}_{ab}, \ \mathcal{S}_{abc}, \ D_c\mathcal{S}_{ab}, \ 
D_d \mathcal{S}_{abc}).
\end{eqnarray*}
\end{subequations}
\end{proof}

\begin{remark}
\em{The system of wave equations in the previous lemma is homogeneous in the
zero-quantities $\mathcal{S}_{ab}, \ \mathcal{S}_{abc}$ and $\mathcal{B}_{ab}$
as long as the \emph{obstruction tensor} $\mathcal{O}_{ab}$ vanishes identically on
$\mathscr{I}$.}
\end{remark}

\begin{remark}
\em{ If $\mathscr{I}$ is conformally flat, then the obstruction tensor
vanishes identically as $y_{abc} = 0$.}
\end{remark}

Lemmas \ref{WaveEquationsKillingVector} and \ref{WaveEquationsZeroQuantities}
lead to the following proposition:

\begin{proposition}
\label{MainPropBoundary}
Let $(\mathcal{M}, g_{ab})$ a conformal extension of an anti-de Sitter-like
spacetime with corner $\partial\mathcal{S} \equiv \mathcal{S}_\star \cap
\mathscr{I}$. Let $\zeta_a$ and $\eta$ fields satisfying
\eqref{KillingVectorBdy} and \eqref{DerDerEtaBdy}, and $y_{abc}$ a tensor
with the symmetries of the magnetic part of the Weyl tensor.
Assume that $\mathcal{S}_{ab}, \ \mathcal{B}_{ab}$ and $\mathcal{S}_{abc}$
vanish identically at $\partial\mathcal{S}$. Then $\zeta_a$ satisfies the
unphysical conformal Killing equation on $\mathscr{I}$ if and only if $
\mathcal{O}_{ab} \simeq 0$.
\end{proposition}

\begin{remark}
{\em We stress that the vanishing of the obstruction tensor $
\mathcal{O}_{ab}$ is a necessary and sufficient condition for
the existence of a Killing vector on the spacetime. The necessity
follows from the fact that if a Killing vector is present in the
spacetime then all the zero-quantities associated to the conformal
Killing vector evolution system will vanish. This, in turn, implies
that the zero-quantities intrinsic to the conformal boundary have to
vanish. The last of the wave equations in Lemma
\ref{WaveEquationsZeroQuantities} implies then that  $\mathcal{O}_{ab}
\simeq 0$. }
\end{remark}

\begin{remark}
{\em It should be stressed that the analysis carried out in the
  previous sections is conformally invariant. More precisely, if the
  unphysical Killing vector candidate is such that the zero-quantities
associated to the Killing equation conformal evolution system
vanish for a particular conformal representation, then it follows that
they will also vanish for any other conformal representation. This
follows from the conformal transformation properties for the
zero-quantities implied by the change of connection transformation
formulae. From this observation it follows also that the reduced
Killing boundary conditions \eqref{EtaBdy}-\eqref{DerDerEtaBdy} have
similar 
conformal invariance properties.}
\end{remark}

\section{Initial data at $\partial\mathcal{S}$}

As mentioned in \Cref{Corner}, the system \eqref{WEBdy1}-\eqref{WEBdy2} must be
complemented with data at $\partial\mathcal{S}$, that is to say, we have to
bring into consideration the conditions implied by the zero-quantities on
$\mathcal{S}_\star$ and make them consistent with the ones obtained from the boundary
analysis in the previous section. The main difference between this section and
the preceding ones is the introduction of an adapted system of coordinates
suited for studying the corner conditions.

\subsection{Set up}

For simplicity, let us introduce a system of coordinates $x^\mu = (x^0, x^1,
x^\mathcal{A})$ where $x^0$ and $x^1$ correspond to the time and radial coordinates,
respectively, while the caligraphic index $\mathcal{A}$ represents angular
coordinates. This system of coordinates is adapted to our problem in
the sense that
$\mathcal{S}_\star$ and $\mathscr{I}$ are given by
\begin{align*}
\mathcal{S}_\star = \{p\in\mathcal{M} \ | \ x^0 = 0 \} \quad \textrm{and} \quad 
\mathscr{I} = \{p\in\mathcal{M} \ | \ x^1 = 0 \}.
\end{align*}
The corner is determined then by the condition
$x^0 = x^1 = 0$.

\medskip

Let $h_{ab}$ be the intrinsic metric on $\mathcal{S}_\star$ and $t^a$ be its
the normal vector.  As the hypersurface $\mathcal{S}_\star$ is spacelike then
$t_a t^a = -1$. For convenience, let use the symbol $\hat{\phantom{X}}$ to
denote quantities defined on this hypersurface. 

Once coordinates have been introduced, the metrics can be written explicitly in
terms of the lapse and shift functions. Adopting a Gaussian gauge, the metrics
on $\mathcal{S}_\star$ and $\mathscr{I}$ take, respectively, the forms
\begin{subequations}
\begin{eqnarray}
&&\bmg |_{\mathcal{S}_\star} = -\mathbf{d} x^0\otimes \mathbf{d} x^0 + h_{\alpha\beta}\mathbf{d} x^\alpha
\otimes\mathbf{d} x^\beta,  \quad (\alpha, \beta = 1, 2, 3) \\
&& \bmg \simeq \mathbf{d} x^1\otimes \mathbf{d} x^1 + \ell_{\gamma\delta}\mathbf{d} x^\gamma \otimes \mathbf{d} x^\delta
\quad (\gamma, \delta = 0,2,3).
\end{eqnarray}
\end{subequations}
From here, we find that the non-zero components of the metric at the corner
$\partial\mathcal{S}$ are:
\[
g_{00} = \ell_{00} = -1, \quad g_{11} = h_{11} = 1, \quad g_{\mathcal{A}\mathcal{B}} =
h_{\mathcal{A}\mathcal{B}} = \ell_{\mathcal{A}\mathcal{B}}.
\]

\subsection{Corner conditions}

As noticed in \Cref{Corner}, the wave equations \eqref{WEBdy1} and
\eqref{WEBdy2} require suitable initial data at $\partial\mathcal{S}$.  These
are naturally provided by the conditions the initial data impose on $\eta, \
\zeta_a$ and their first derivatives along the conformal boundary. Here we
describe how such conditions can be obtained.

Let $\hat{\zeta}_a$ and $\hat{\zeta}$ denote, respectively, the pull-backs of
$h_a{}^b\xi_b$ and $t^a\xi_a$ into $\mathcal{S}_\star$.  Although this
decomposition with respect to $h_{ab}$ is clearly different from the one
performed on the conformal boundary we can observe that, when expressed in the
adapted coordinates $x^\mu$, the following identities hold at the corner:
\[
\hat{\zeta}_1 = \zeta = 0, \qquad \hat{\zeta} = \zeta_0, \qquad
\hat{\zeta}_\mathcal{A} = \zeta_\mathcal{A}.
\]
In this way, the angular components $\zeta_\mathcal{A}$ on $\mathscr{I}$ are
fixed by the initial data.  Similarly, if one requires the conformal factor
$\Xi$ to have continuous first derivatives, it follows then that the conditions
\[
\hat\partial_0 \Xi =\partial_0\Xi = 0, \quad \hat\partial_1\Xi = \partial_1\Xi = \Sigma, \quad
\hat\partial_\mathcal{A}\Xi = \partial_\mathcal{A}\Xi = 0.
\]
must be satisfied at $\partial\mathcal{S}$.

\medskip 

Regarding the remaining fields, values for $\eta$ and the
components of $\xi_a$ on $\mathcal{S}_\star$ can be found solving equations
\eqref{DecompPhi}--\eqref{DecompDerPsi} and
\eqref{DecompSab}--\eqref{DecompDBab} --the KID equations set-- with $\epsilon
=-1$.  Moreover, this system also provides with all their derivatives. In
particular, when the limit $\Xi \to 0$ is taken, the corresponding solutions
for $\eta, \ \hat{\zeta}_0$ and $\hat{\zeta}_\mathcal{A}$ along with their time
and angular derivatives serve as initial data at $\partial\mathcal{S}$ for wave
equations \eqref{BoxXi} and \eqref{BoxEta}.

\section{Conclusions} Once the conditions for the existence of a conformal
Killing vector on $\mathscr{I}$ have been established, we can link
\Cref{MainPropBoundary} to the initial-boundary problem in the spacetime via
Lemmas \ref{WEKillingSpacetime} and \ref{WEZQ}.  The main result of this work
can be formulated as follows:

\begin{theorem}
\label{MainTheorem}
Let $(\mathcal{M}, \bmg)$ a conformal extension of an anti de Sitter-like
spacetime with conformal boundary $\mathscr{I}$.  Let $\mathcal{S}_\star
\subset \mathcal{M}$ be a spacelike hypersurface intersecting $\mathscr{I}$ at
$\partial\mathcal{S}$. Let $\xi_{a\star}$ and $\eta_\star$ satisfy the
conformal KID equations \eqref{ConformalKID1} and \eqref{ConformalKID2} on
$\mathcal{S}_\star$.  Let $\zeta_a$ and $\eta$ be the fields obtained from
solving the wave equations \eqref{WEBdy1} and $\eqref{WEBdy2}$ with initial
data given by the restriction of $\xi_{a\star}$ and $\eta_\star$ to
$\partial\mathcal{S}$. Assume further that the obstruction tensor
$\mathcal{O}_{ab}$ constructed from $\ell_{ab}, \ \eta$ and $\zeta_a$ and
defined by equation \eqref{Obstruction} vanishes. Then the Killing vector
candidate $\xi_a$ obtained from solving equations \eqref{BoxXi} and
\eqref{BoxEta} with initial data $\xi_{a\star}, \ \eta_{\star}$ and boundary
data $\zeta_a , \ \eta$ pull-backs to a Killing vector $\tilde{\xi}_a$.
\end{theorem}

\begin{remark}
{\em The obstruction tensor $\mathcal{O}_{ab}$  clearly vanishes for
  conformally flat boundaries. The question remains, however, 
  whether there exist other conformal classes of Lorentzian metrics
  with this property. Addressing this question may require expanding
  the obstruction tensor in a particular gauge with the aim of finding
  explicit solutions to this condition. This interesting question is,
  however, outside the scope of this article and will be pursued elsewhere.}
\end{remark}

\section*{Acknowledgements}
The authors thank the hospitality of the International Erwin Schr\"odinger
Institute for Mathematics and Physics where part of this work was carried out
as part of the research programme \emph{Geometry and Relativity} during
July-September 2017. DAC thanks support granted by CONACyT (480147). The
calculations in this article have been carried out in the suite {\tt xAct} for
abstract tensorial manipulations ---see \cite{xAct}.



\begin{thebibliography}{10}

\bibitem{BeiChr97b}
R.~Beig \& P.~T. Chru\'{s}ciel,
\newblock {\em Killing initial data},
\newblock Class. Quantum Grav. {\bf 14}, A83 (1997).

\bibitem{CheWah83}
C.~Chen \& W.~von Wahl,
\newblock {\em Das Rand-Anfangswertproblem f{\"u}r quasilineare
  Wellengleichungen in Sobolevr\"aumen niedriger Ordnung},
\newblock J. Reine Angew. Math. {\bf 337}, 77 (1983).

\bibitem{Chr91b}
P.~T. Chru\'sciel,
\newblock {\em On the uniqueness in the large of solutions of Einstein's
  equations ("Strong Cosmic Censorship")},
\newblock Centre for Mathematics and its Applications, Australian National
  University, 1991.

\bibitem{ChrPae13a}
P.~T. Chru\'{s}ciel \& T.-T. Paetz,
\newblock {\em KIDs like cones},
\newblock Class. Quantum Grav. {\bf 30}, 235036 (2013).

\bibitem{DafHru85}
C.~M. Dafermos \& W.~J. Hrusa,
\newblock {\em Energy methods for quasilinear hyperbolic initial-boundary value
  problems. Applications to elastodynamics.},
\newblock Arch. Rational Mech. Analysis {\bf 87}, 267 (1985).

\bibitem{Fri95}
H.~Friedrich,
\newblock {\em {Einstein} equations and conformal structure: existence of
  anti-de {Sitter}-type space-times},
\newblock J. Geom. Phys. {\bf 17}, 125 (1995).

\bibitem{Fri09}
H.~Friedrich,
\newblock {\em Initial boundary value problems for Einstein's field equations
  and geometric uniqueness},
\newblock Gen. Rel. Grav. {\bf 41}, 1947 (2009).

\bibitem{FriNag99}
H.~Friedrich \& G.~Nagy,
\newblock {\em The Initial Boundary Value Problem for Einstein's Vacuum Field
  Equation},
\newblock Comm. Math. Phys. {\bf 201}, 619 (1999).

\bibitem{HolSha16}
G.~Holzegel \& A.~Shao,
\newblock {\em Unique Continuation from Infinity in Asymptotically Anti-de
  Sitter Spacetimes},
\newblock Comm. Math. Phys. {\bf 374}, 723 (1916).

\bibitem{HolSha16b}
G.~Holzegel \& A.~Shao,
\newblock {\em Unique continuation from infinity in asympotically Anti-de
  Sitter spacetimes II: Non-static boundaries},
\newblock in {\tt arXiv1608.07521}, 2016.

\bibitem{xAct}
J.~M. Mart\'{\i}n-Garc\'{\i}a,
\newblock http://www.xact.es, 2014.

\bibitem{Mon75}
V.~Moncrief,
\newblock {\em Spacetime symmetries and linearization stability of the Einstein
  equations. I.},
\newblock J. Math. Phys. {\bf 16}, 493 (1975).

\bibitem{Pae14}
T.-T. Paetz,
\newblock {\em KIDs prefer special cones},
\newblock Class. Quantum Grav. {\bf 31}, 085007 (2014).

\bibitem{Pae15}
T.-T. Paetz,
\newblock {\em Conformally covariant systems of wave equations and their
  equivalence to Einstein's field equations},
\newblock Ann. Henri Poincar\'e {\bf 16}, 2059 (2015).

\bibitem{Pae16}
T.-T. Paetz,
\newblock {\em Killing Initial Data on spacelike conformal boundaries},
\newblock J. Geom. Phys. {\bf 106}(51) (2016).

\bibitem{CFEBook}
J.~A. {Valiente Kroon},
\newblock {\em Conformal Methods in General Relativity},
\newblock Cambridge University Press, 2016.

\end{thebibliography}

\end{document}